\documentclass[10pt, draftclsnofoot, onecolumn]{IEEEtran}

\usepackage{amsmath}
\usepackage{amssymb}
\usepackage{amsthm}
\usepackage{cite}

\usepackage{graphicx}
\usepackage{enumerate}
\usepackage{setspace}
\usepackage{subfigure}
\usepackage{color}

\newtheorem{theorem}{Theorem}

\newtheorem{lemma}{Lemma}
\newtheorem{corollary}{Corollary}

\title{Simple transmission strategies for interference channel}

\author{Jung Hyun Bae, Jungwon Lee, Inyup Kang\\
Mobile Solutions Lab\\
Samsung US R$\&$D Center\\
San Diego, CA, USA\\
Email: jungbae@umich.edu, jungwon@alumni.stanford.edu, inyup.kang@samsung.com}

\begin{document}
\maketitle
\begin{abstract}
In this paper, we investigate performances of simple transmission strategies. We first consider two user SISO Gaussian symmetric interference channel (IC) for which Etkin, Tse and Wang proposed a scheme (ETW scheme) which achieves one bit gap to the capacity. We compare performance of point-to-point (p2p) codes with that of the ETW scheme in practical range of transmitter power. It turns out that p2p coding scheme performs better or as nearly good as the ETW scheme. Next, we consider K user SISO Gaussian symmetric IC. We define interference regimes for K user SISO Gaussian symmetric IC and provide closed-form characterization of the symmetric rate achieved by the p2p scheme and the ETW scheme. Using this characterization, we evaluate performances of simple strategies with K=3, and show the similar trend to two user case. 
\end{abstract}
\section{Introduction}
In wireless communication, managing interference has crucial importance for reliable communication due to its nature of shared communication medium. Given presence of interference, a receiver can either treat interference as noise (IAN receiver) or try to decode interference (joint decoding receiver) even though it is not ultimately interested in decoding interference. For an IAN receiver, larger interference would always result in higher noise floor, so reducing interference would be a good transmission strategy. This strategy has been well adopted in the traditional wireless system design in the form of orthogonal multiplexing schemes. When a receiver has multiple dimensions of observation, aligning interference into sub-dimensions was shown to achieve the optimal degrees of freedom (DOF) which is strictly greater than that of orthogonal multiplexing schemes for certain cases~\cite{CaJa08}. Still, the focus of the transmission strategy is to provide interference-free dimension by assuming an IAN receiver.\\
\indent
Although aforementioned interference-alignment has brought surprisingly good result with the assumption of an IAN receiver, it turns out that a receiver needs to decode interference in some degree if one is interested in the optimal achievable rate, i.e. capacity. When a receiver is willing to decode interference, reducing interference or providing interference-free dimension is not necessarily a good strategy. Han and Kobayashi looked at this problem more than 30 years ago and proposed well-known Han-Kobayashi (HK) scheme for two user interference channel (IC)~\cite{HaKo81}. In the HK scheme, message of each transmitter is divided into common and private part. It is the common part of message which needs to be decoded at the both receivers. Finding capacity for this simple two user IC is still an open problem, and the HK scheme is the best known strategy. Recently Etkin, Tse and Wang considered two user single-input, single-output (SISO) Gaussian IC, and proposed a scheme (ETW scheme) which achieves one bit gap to the capacity~\cite{EtTsWa08}. The ETW scheme is a simpler version of the HK scheme with fixed common-private splitting, but is still shown to be close to the capacity.\\
\indent
Aforementioned result for two user SISO Gaussian IC was generalized for two user multiple-input, multiple-output (MIMO) Gaussian IC in~\cite{KaVa11-1,PaBlTa08}, and it was shown that the ETW-like scheme achieves the constant gap to the capacity. When there are more than two tranceiver pairs, there are very few capacity-like results. It is shown that the ETW-like scheme is generalized DOF (GDOF) optimal for certain classes of $K$ user MIMO Gaussian IC~\cite{GoJa11, MoMu11}. In $K$ user SISO Gaussian IC, signal level alignment is shown to achieve GDOF of symmetric case. This signal level alignment was also used to show GDOF result of many-to-one IC~\cite{BrPaTs10} and two user X-channel~\cite{HuCaJa08}. \\
\indent
Aforementioned capacity-like results are based on schemes which require coordination of transmitters and/or knowledge of the interfering channel at each transmitter. Baccelli et al. investigated capacity of $K$ user SISO Gaussian IC when there is neither transmitter coordination nor knowledge of the interfering channel~\cite{BaGaTs11}. It can be easily seen that this strategy is not even GDOF optimal, and hence, it exhibits infinite gap to the capacity as transmit power goes to infinity. This implies that the optimal coding scheme for IC could be significantly different from the optimal coding scheme for p2p channel, and this aspect is explored in~\cite{BeShCa10} to show that ``bad" LDPC codes for p2p channel can be ``good" codes for IC. Nevertheless, performance of JD receiver is shown to be considerably better than IAN receiver even when p2p-capacity-achieving codes are used~\cite{BaGaTs11}. \\
\indent
Motivated by this result, we investigate performances of ``simple" transmission strategies in this paper. We first consider two user SISO Gaussian symmetric IC. In this case, the ETW scheme is already shown to be near-optimal by its one bit gap to the capacity. Surprisingly, it turns out that the p2p coding scheme performs better or as nearly good as the ETW scheme.  Next, we consider $K$ user SISO Gaussian symmetric IC. When there are more than two users, the HK or the ETW scheme is known to be not even GDOF optimal. Unfortunately, the only capacity-like result for more than two user case is obtained only for symmetric IC using signal level alignment as mentioned earlier, and practically implementing this signal level alignment would be quite challenging. Therefore, evaluating performance of simple strategies could have significant meaning for practical purpose. In this paper, we define interference regimes for $K$ user SISO Gaussian symmetric IC and provide closed-form characterization of the symmetric rates achieved by the p2p coding scheme and by the ETW scheme. Using this characterization, we evaluate performances of simple strategies with $K=3$, and show that the p2p coding scheme still performs well with respect to the ETW scheme for practical SNR range.\\
\indent
The rest of the paper is organized as follows. In section~\ref{sec:sym}, performance evaluation of simple strategies along with characterization of the symmetric rate is done for two user SISO Gaussian symmetric IC. Characterization of the sum rate with p2p codes is also done for asymmetric case as part of analysis. In section~\ref{sec:k}, characterization of the symmetric rates with p2p codes and with the ETW scheme is done, and performance evaluation is provided. Section~\ref{sec:con} concludes the paper. 
\section{Two user SISO Gaussian symmetric IC}
\label{sec:sym}
\subsection{Interference regimes and achievable region of p2p-capacity-achieving codes}
Let us define two user SISO Gaussian symmetric IC with channel inputs $X_1$, $X_2$ and channel outputs $Y_1$, $Y_2$ as follows.
\begin{eqnarray}
Y_1&=&\sqrt{P}X_1+\sqrt{aP}X_2+Z_1\\
Y_2&=&\sqrt{aP}X_1+\sqrt{P}X_2+Z_2,
\end{eqnarray}
where $a,P>0$ and $Z_1, Z_2 \sim \mathcal{CN}(0,1)$. Note that $P$ represents signal-to-noise ratio (SNR), and $a$ represents interference-to-signal ratio (ISR).
Because the channel is symmetric, we have 
\begin{subequations}
\begin{eqnarray}
I(X_1;Y_1)&=& I(X_2;Y_2),\\
I(X_1;Y_1|X_2)&=& I(X_2;Y_2|X_1),\\
I(X_1;Y_2)&=& I(X_2;Y_1),\\
I(X_1;Y_2|X_2)&=& I(X_2;Y_1|X_1).
\end{eqnarray}
\end{subequations}
For this channel, we would like to define four interference regimes according to interference level. In noisy interference regime,
$I(X_1;Y_1)\geq I(X_1;Y_2|X_2)$ holds. Roughly speaking, this implies that interference-affected version of the direct link ($X_1$ to $Y_1$) is better than interference-free version of the cross link($X_1$ to $Y_2$). In Gaussian channel, noisy interference regime corresponds to the range of $a\leq \frac{-1+\sqrt{1+4P}}{2P}$. In weak interference regime, $I(X_1;Y_1|X_2)\geq I(X_1;Y_2|X_2)>I(X_1;Y_1)$ holds which implies that interference-free version of the direct link is better than interference-free version of the cross link. In Gaussian channel, this corresponds to the range of $\frac{-1+\sqrt{1+4P}}{2P}<a\leq 1$. In strong interference regime, $I(X_1;Y_2|X_2)> I(X_1;Y_1|X_2)\geq I(X_1;Y_2)$ holds which implies that interference-free version of the cross link is better than interference-free version of the direct link. In Gaussian channel, this corresponds to the range of $1<a\leq 1+P$. In very strong interference regime, $I(X_1;Y_2)> I(X_1;Y_1|X_2)$ holds which implies that interference-affected version of the cross link is better than interference-free version of the direct link. In Gaussian channel, this corresponds to the range of $a>1+P$.\\
\indent
As in~\cite{BaGaTs11}, we define p2p-capacity-achieving codes as length $n$ block codes which achieves a rate of $R$ over every p2p Gaussian channel with capacity greater than $R$ as $n \rightarrow \infty$. As mentioned earlier, p2p-capacity-achieving codes excludes HK or ETW schemes. In~\cite{BaGaTs11}, capacity region with p2p-capacity-achieving codes is characterized. For noisy interference regime, it is given as union of $\mathcal{C}_0$ and $\mathcal{C}_1$ which are given as 
\begin{subequations}
\begin{eqnarray}
\mathcal{C}_0&=&\left\{(R_1,R_2): R_1 < I(X_1;Y_1), R_2 < I(X_2;Y_2)\right\},\\
\mathcal{C}_1&=&\{(R_1,R_2): R_1 < I(X_1;Y_1|X_2), R_2 < I(X_2;Y_2|X_1), R_1+R_2<I(X_1,X_2;Y_2) \}.
\end{eqnarray}
\end{subequations}
Note that $\mathcal{C}_0$ corresponds to the achievable region of IAN receiver, and $\mathcal{C}_1$ corresponds to the achievable region of simultaneous decoding receiver which is defined in~\cite{GaKi10} and used in~\cite{BaGaTs11}. For weak, strong and very strong interference regimes, the capacity region is equal to $\mathcal{C}_1$, which means treating interference as noise is meaningful only if interference is very weak. For very strong interference regime, the sum-rate bound in $\mathcal{C}_1$ is actually ineffective, i.e. decodability of intended message without interference is limiting factor because interference decoding is easy due to strong interference, and the capacity region is given as $\mathcal{C}'_1=\{(R_1,R_2): R_1 < I(X_1;Y_1|X_2), R_2 < I(X_2;Y_2|X_1) \}$. For very strong interference regime, the capacity of each user is the same as no interference case. It has also been shown that the capacity region of p2p-capacity-achieving codes are equal to the capacity region in strong and very strong interference regimes~\cite{Ca75, Sa78, HaKo81}. One might think that decodability of unintended message without interference (caused by intended message) must be a limiting factor in noisy interference regime. In that case, the capacity region of p2p-capacity-achieving codes would reduce to $\mathcal{C}_0$. This does not happen because we consider simultaneous decoding receiver defined in~\cite{GaKi10}. This receiver has essentially the same decoding procedure as the traditional joint decoding receiver. The only difference is that an error event corresponding to the error of only unintended message is not included when evaluating the achievable region. Figure~\ref{fig:noisy} explains this phenomenon, and more detailed explanation can be found in~\cite{GaKi10}. Figure~\ref{fig:other} shows achievable regions of several aforementioned receivers in weak, strong, and very strong regimes.
\begin{figure}[t]
  \subfigure[Achievable region at receiver 1]{\label{fig:rec1}\includegraphics[width=0.33\textwidth]{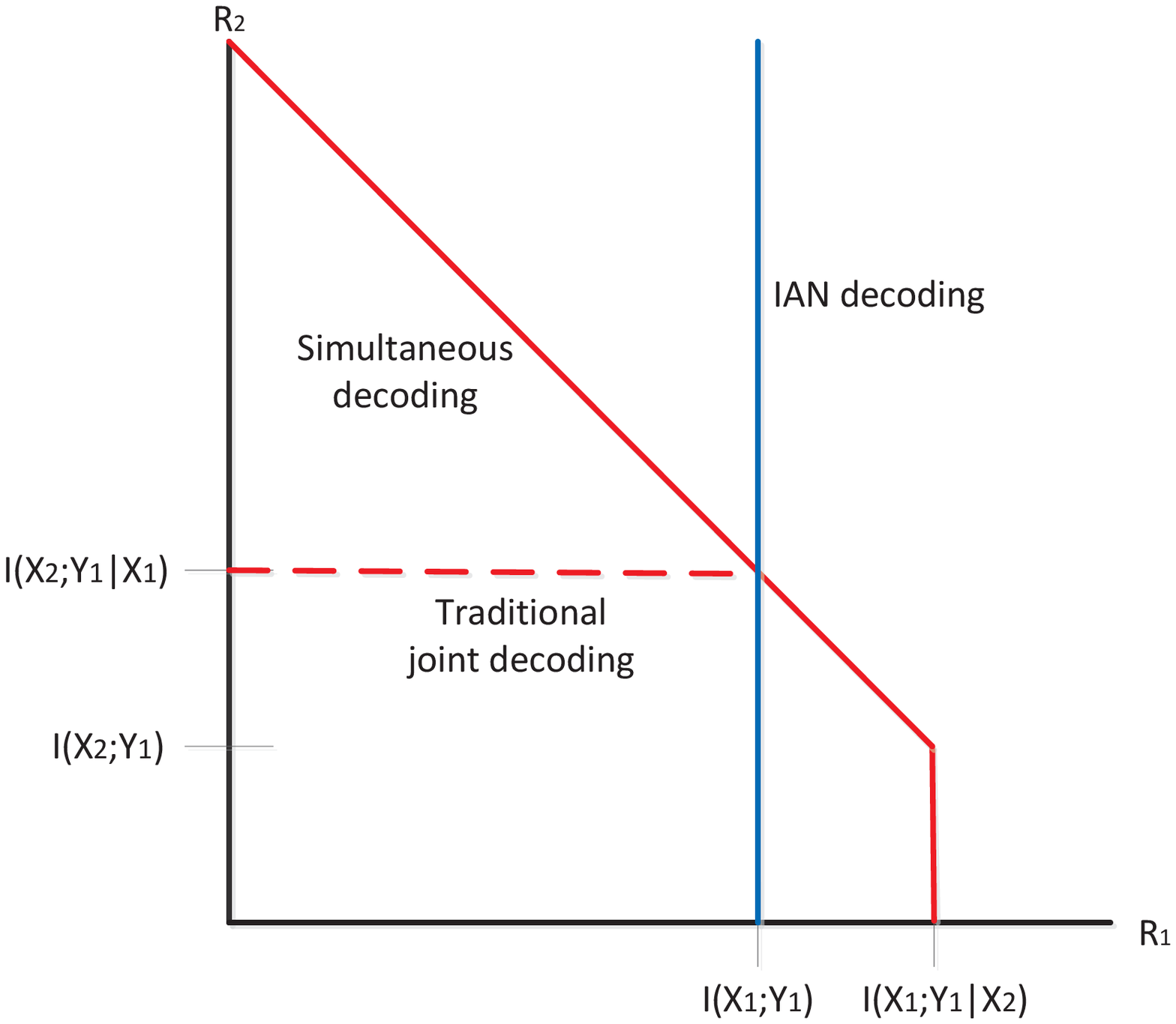}} 
  \subfigure[Achievable region at receiver 2]{\label{fig:rec2}\includegraphics[width=0.33\textwidth]{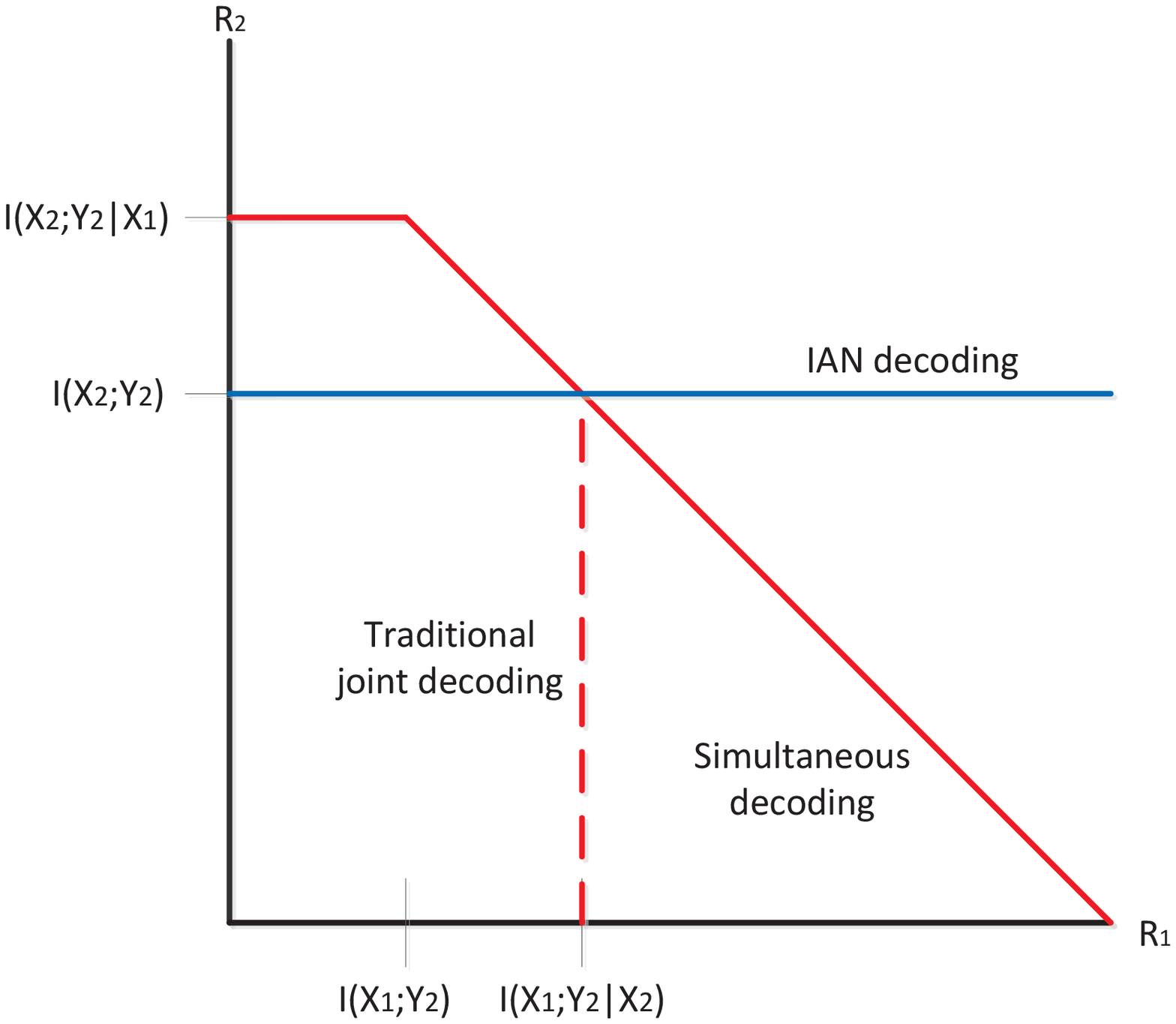}} 
  \subfigure[Overall achievable region]{\label{fig:all}\includegraphics[width=0.33\textwidth]{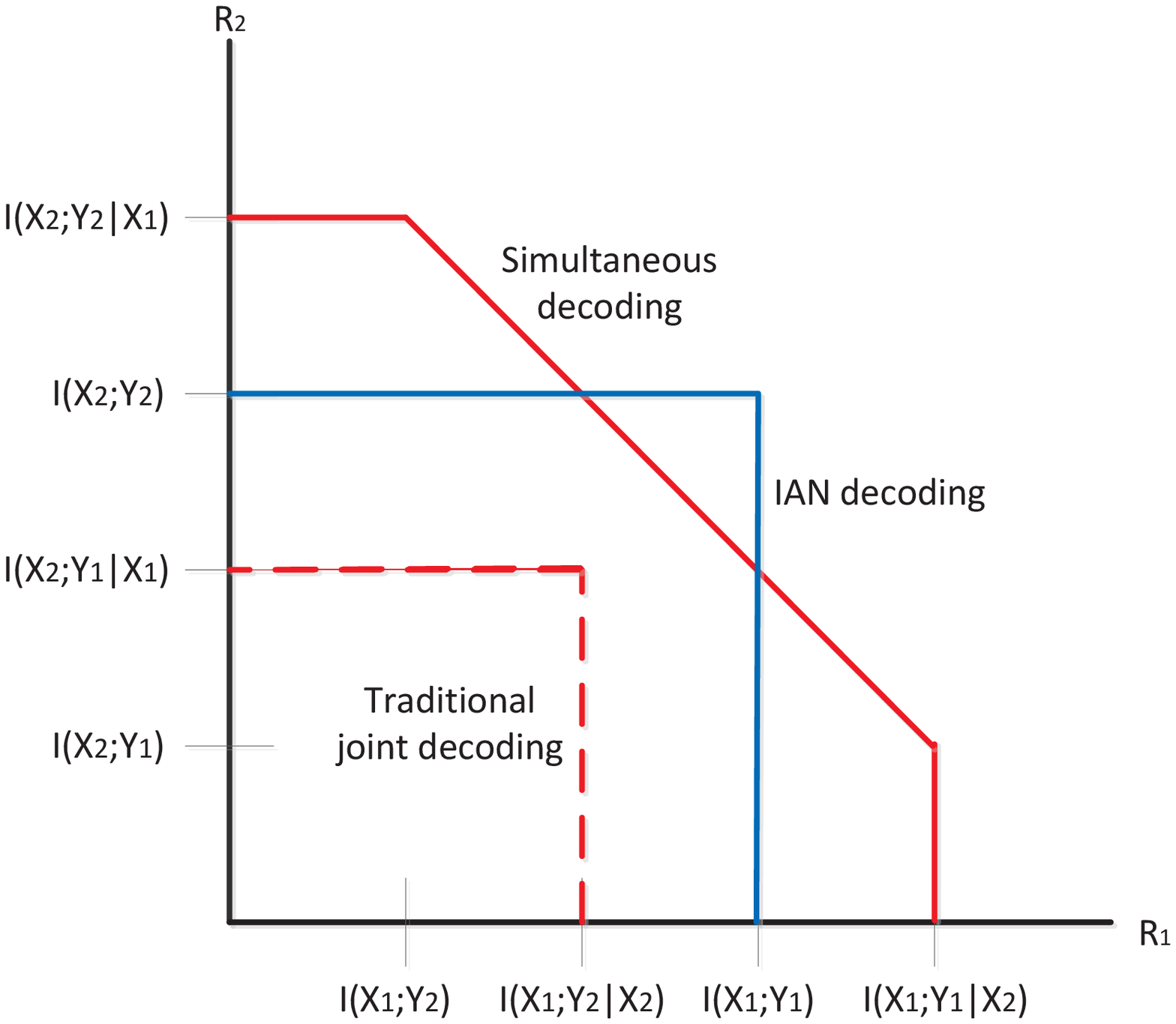}} 
  \caption{Achievable region in noisy interference regime}
  \label{fig:noisy}
\end{figure}
\begin{figure}[t]
  \subfigure[Weak interference regime]{\label{fig:weak}\includegraphics[width=0.33\textwidth]{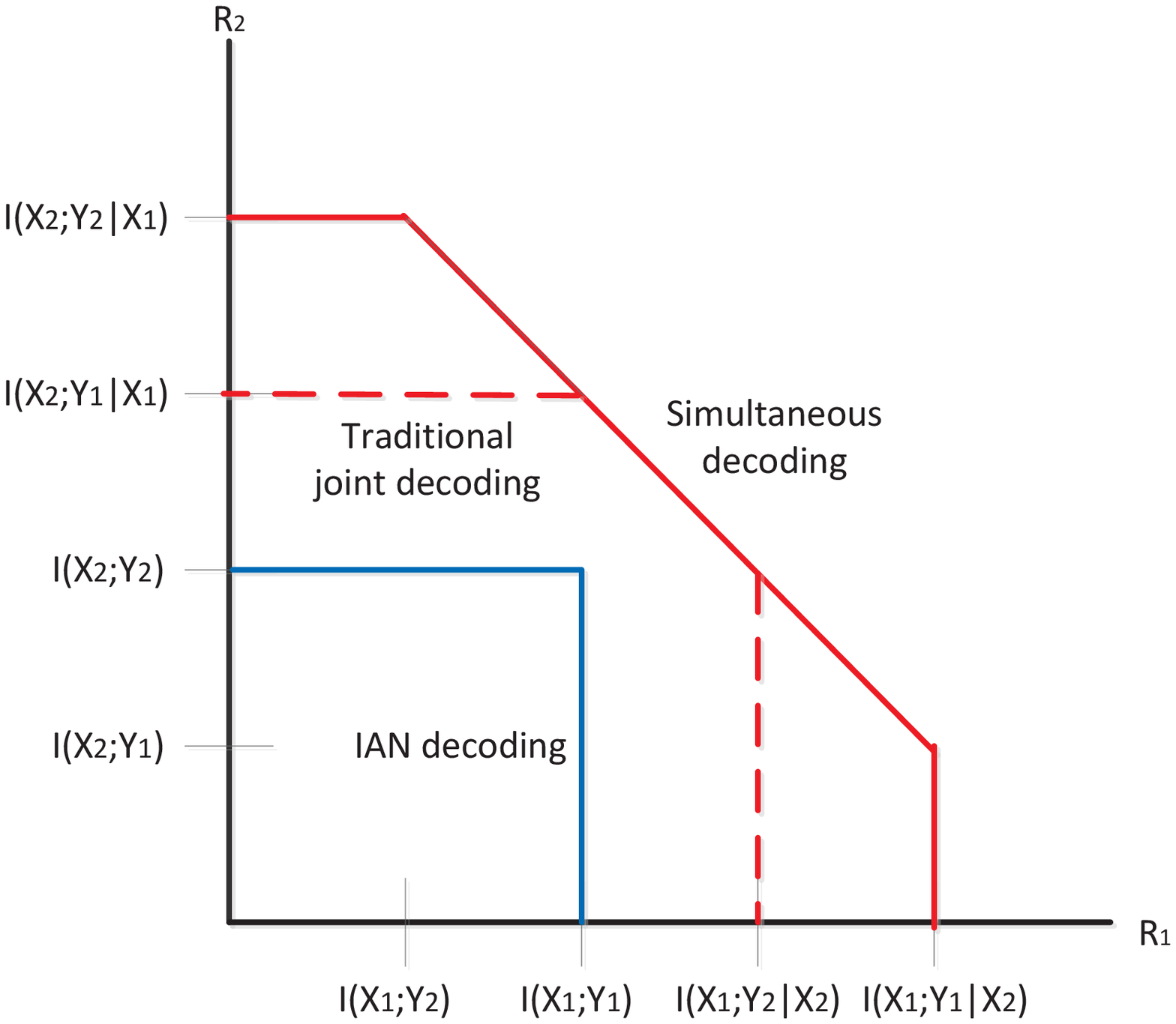}} 
  \subfigure[Strong interference regime]{\label{fig:strong}\includegraphics[width=0.33\textwidth]{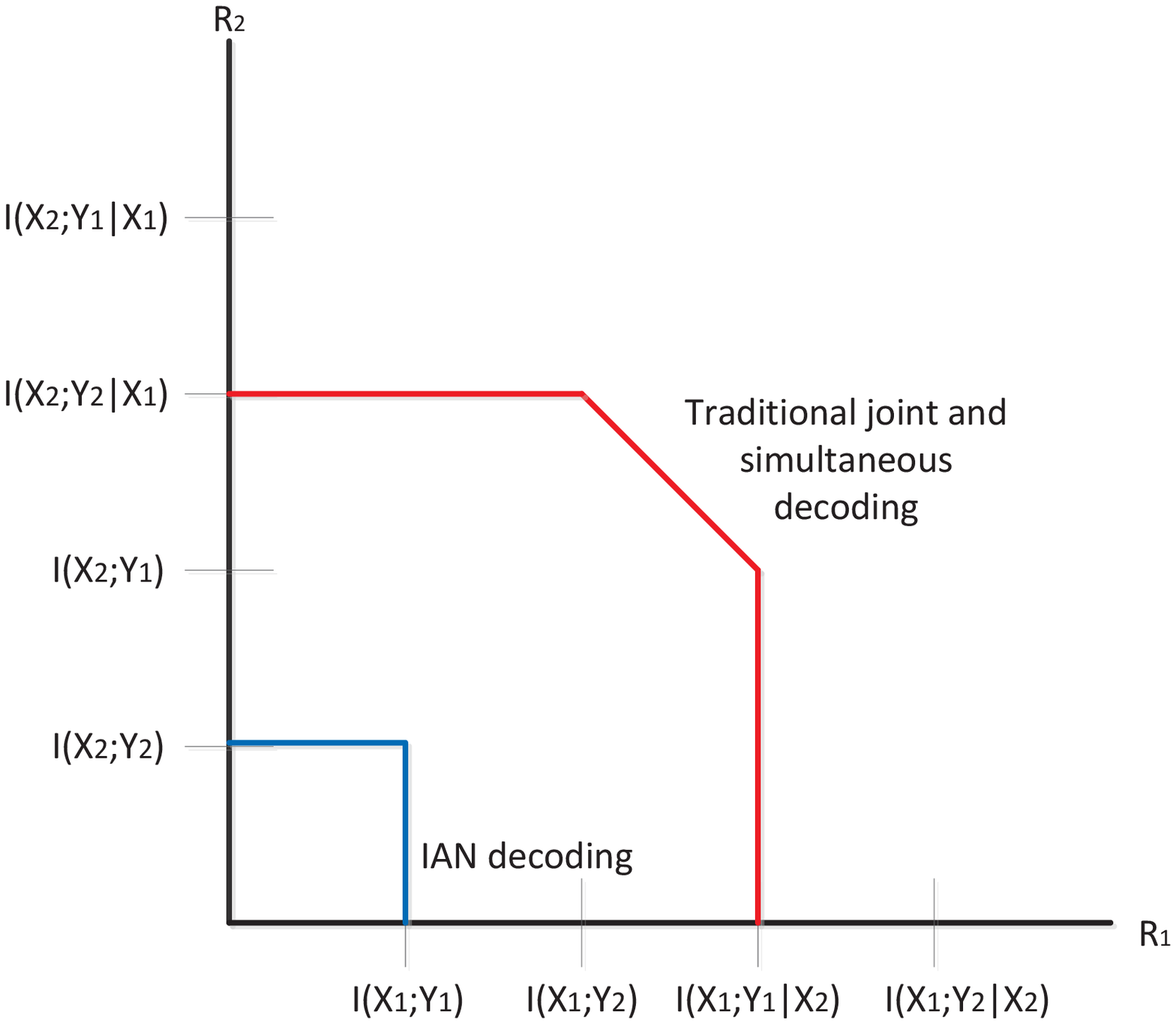}} 
  \subfigure[Very strong interference regime]{\label{fig:very}\includegraphics[width=0.33\textwidth]{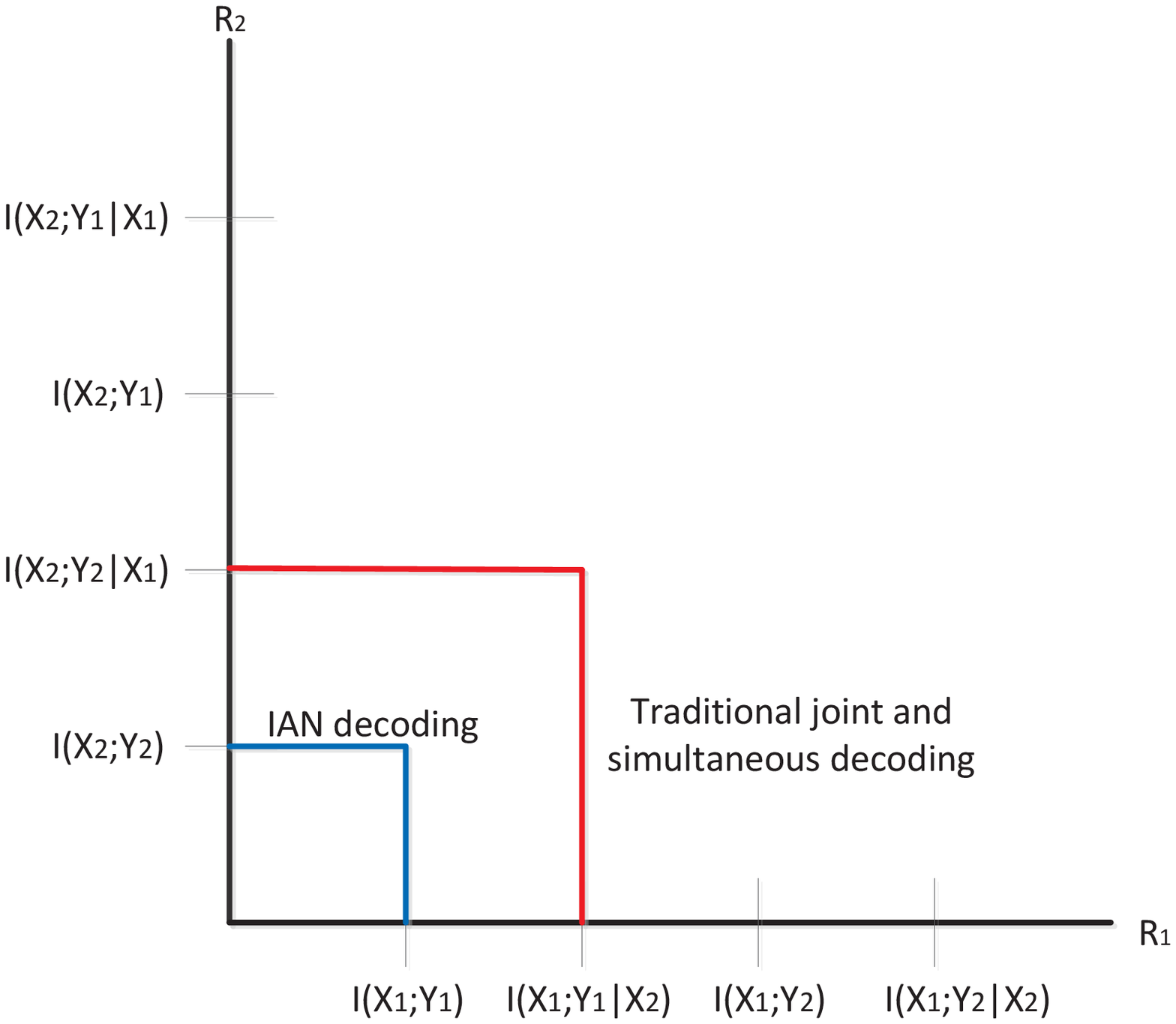}} 
  \caption{Achievable regions for weak, strong and very strong interference regimes}
  \label{fig:other}
\end{figure}
\subsection{Performance comparison of simple transmission schemes}
Let us define the symmetric rate $(C_{sym})$ of a scheme as $ C_{sym}=\max_{(R_1,R_2) \in \mathcal{R}} \min \{R_1,R_2\}$, where $\mathcal{R}$ is the achievable region of a scheme. It is known that the ETW scheme achieves one bit gap to the symmetric capacity. In this section, we would like to compare performance of an even simpler scheme with that of the ETW scheme. Since p2p-capacity-achieving codes are known to achieve the capacity region of strong and very strong interference regimes, we only consider noisy and weak interference regimes. Achievable region of a ``simpler'' scheme considered in this section is the union of the aforementioned capacity region of the p2p-capacity-achieving codes and the achievable region of a TDMA scheme in which only one of two users transmit at each time. For convenience, let us call this unified scheme as the ``p2p scheme''. We may also consider a scheme in which one user deliberately uses less than full power to reduce interference to the other user. Although SNR of this user gets worse by doing this, its achievable rate may not be far worse because ISR for this user gets larger which may be beneficial. In this case, the symmetric rate of the scheme will still be limited by this user's rate, but we may get better sum rate. It turns out that this is not beneficial even in terms of the sum rate as will be seen later, and hence this scheme need not be considered. It is well known that a TDMA scheme is far from GDOF optimal, and hence may be thought as a ``bad'' scheme. Surprisingly, it will be seen that it is difficult to outperform a TDMA for practical SNR range in weak interference regime. For simplicity, performance comparison will be done in terms of the symmetric rate. Since the TDMA scheme achieves the symmetric rate of $\frac{1}{2}\log_2(1+2P)$ regardless of interference regime, it is easy to see that the symmetric rate with of the p2p scheme is given as follows.
\begin{itemize}
\item \textbf{The symmetric rate of the p2p scheme } \\
\begin{enumerate}
\item Noisy interference regime ($0<a\leq\frac{-1+\sqrt{1+4P}}{2P}$)
\begin{equation}
 C^{p2p}_{sym}=\max \bigg\{ \log_2{\Big(1+\frac{P}{1+aP} \Big)}, \frac{1}{2}\log_2(1+2P)   \bigg\}.
 \end{equation}
\item Weak interference regime ($\frac{-1+\sqrt{1+4P}}{2P}<a\leq 1$)
\begin{equation}
C^{p2p}_{sym}=\max \bigg\{ \frac{1}{2}\log_2(1+P+aP), \frac{1}{2}\log_2(1+2P)   \bigg\}= \frac{1}{2}\log_2(1+2P).
\end{equation}
\end{enumerate}
\end{itemize}
The symmetric rate of the ETW scheme is given in~\cite{EtTsWa08} as follows. 
\begin{itemize}
\item \textbf{The symmetric rate of the ETW scheme} \\
\begin{eqnarray}
\label{eq:etw}
C^{ETW,SIC}_{sym}=\begin{cases}  \log_2{\Big(1+\frac{P}{1+aP}\Big)}, & \quad a \leq \frac{1}{P} \\  \min\bigg\{\frac{1}{2}\log_2(1+P+aP)+ \frac{1}{2}\log_2\Big(2+\frac{1}{a} \Big)-1,\log_2\Big(1+aP+\frac{1}{a}\Big)-1          \bigg\}, & \quad  \frac{1}{P} < a \leq 1.     \end{cases}
\end{eqnarray} 
\end{itemize}
From now on we call the symmetric rate of the ETW scheme as $C^{ETW}_{sym}$. Now we are ready to compare performance of the p2p scheme and the ETW scheme in terms of the symmetric rate. First, let us find a range of $a$ in which the first term in minimization for $\frac{1}{P} < a \leq 1$ in~\eqref{eq:etw} is active ($\mathcal{B}_1$) or the second term is active ($\mathcal{B}_2$) given the value of $P$. Define
\begin{equation}
f(a)=Pa^3+a^2-a-1.
\end{equation}
\begin{lemma}
\label{lem:peak}
\begin{equation}
C^{ETW}_{sym}=\begin{cases}  \log_2\Big(1+aP+\frac{1}{a}\Big)-1, & \quad \frac{1}{P}\leq a\leq a_0 \\  \frac{1}{2}\log_2(1+P+aP)+ \frac{1}{2}\log_2\Big(2+\frac{1}{a} \Big)-1, & \quad  a_0<a\leq 1 ,    \end{cases}
\end{equation} 
where $a_0$ is the unique positive real root of $f(a)=0$. 
\end{lemma}
\begin{proof}
We would like to find range of $a$ which satisfies the following.
\begin{subequations}
\begin{eqnarray}
\frac{1}{2}\log_2(1+P+aP)+ \frac{1}{2}\log_2\Big(2+\frac{1}{a} \Big)-1 &>& \log_2\Big(1+aP+\frac{1}{a}\Big)-1 \\
&\Updownarrow&\nonumber\\
P^2a^4-(1+P)a^2+(1-P)a+1 &<& 0\\
&\Updownarrow&\nonumber\\
(Pa-1)(Pa^3+a^2-a-1) &<& 0\\
&\Updownarrow&\nonumber\\
Pa^3+a^2-a-1 &<& 0,
\end{eqnarray}
\end{subequations}
where we use the fact that $a \geq \frac{1}{P}$. Note that $f(a)=Pa^3+a^2-a-1 $. Let us look more closely at $f(a)$. By finding $a$ satisfying $f'(a)=0$, we can find two critical points of $f(a)$ as $\frac{-1-\sqrt{1+3P}}{3P}<0$ and $\frac{-1+\sqrt{1+3P}}{3P}>0$. Note that $\frac{-1+\sqrt{1+3P}}{3P}$ corresponds to the local minimum, and $f(0)=-1<0$. From these facts we can conclude that $f(a)<0$ when $a<a_0$, and $f(a)>0$ when $a>a_0$, where $a_0$ is the only positive real root of $f(a)=0$ which is guaranteed to exist.
\end{proof}
From Lemma~\ref{lem:peak}, we can define
\begin{subequations}
\begin{eqnarray}
\mathcal{B}_1&=&\{a: a_0<a\leq 1\}=\{a: f(a) >0 \}\\
\mathcal{B}_2&=&\{a:  \frac{1}{P}< a\leq a_0   \}=\{a: f(a) \leq 0 \}.
\end{eqnarray}
\end{subequations}
Lemma~\ref{lem:peak} intuitively makes sense because the first term in minimization for $\frac{1}{P} < a \leq 1$ in~\eqref{eq:etw} is related to the sum rate constraint of two common messages which would be active when interference level is high enough while the second term is related to the individual rate constraint of common message which would be active when interference level is low as discussed in~\cite{EtTsWa08}. Now let us state the result for noisy interference regime.
\begin{theorem}
\label{thm:noisy}
In noisy interference regime, i.e. $0<a\leq \frac{-1+\sqrt{1+4P}}{2P}$, the following holds for $P>0$.
\begin{eqnarray}
&C^{p2p}_{sym}&=\begin{cases}  \log_2{\Big(1+\frac{P}{1+aP} \Big)}, & \quad 0< a\leq \frac{-1+\sqrt{1+2P}}{2P} \\  \frac{1}{2}\log_2(1+2P), & \quad  \frac{-1+\sqrt{1+2P}}{2P}<a\leq \frac{-1+\sqrt{1+4P}}{2P} ,    \end{cases}\\
&C^{ETW}_{sym}&\leq  \log_2{\Big(1+\frac{P}{1+aP} \Big)}.
\end{eqnarray}
\end{theorem}
\begin{proof}
In this regime, the symmetric rate of the p2p scheme is lower bounded as 
\begin{equation}
C^{p2p}_{sym}=\max \bigg\{ \log_2{\Big(1+\frac{P}{1+aP} \Big)}, \frac{1}{2}\log_2(1+2P)   \bigg\} \geq \log_2{\Big(1+\frac{P}{1+aP} \Big)}.
\end{equation} 
Note that $\frac{1}{2}\log_2(1+2P)> \log_2{\Big(1+\frac{P}{1+aP} \Big)}$ if $a>\frac{-1+\sqrt{1+2P}}{2P}$, which proves the first claim. \\
\indent
Let us now prove the second claim. If $\frac{1}{P} \geq \frac{-1+\sqrt{1+4P}}{2P} \Leftrightarrow P \leq 2$, then the claim is trivially true from~\eqref{eq:etw}. Let us look at the case where $P>2$. Note that $\mathcal{B}_2$ must have non-empty intersection with $\{a: 0<a\leq \frac{-1+\sqrt{1+4P}}{2P}\}$ in this case. We first need to determine if $\mathcal{B}_1$ has non-empty intersection with $\{a: 0<a\leq \frac{-1+\sqrt{1+4P}}{2P}\}$. Let $a_1=\frac{-1+\sqrt{1+4P}}{2P}$. Then, $Pa_1^2+a_1-1=0  \Leftrightarrow Pa_1^3=a_1-a_1^2$. Therefore, $f(a_1)=Pa_1^3+a_1^2-a_1-1=a_1-a_1^2+a_1^2-a_1-1=-1<0$, which means that $a_1 \in \mathcal{B}_2$. \\
\indent 
Now it suffices to show that $\log_2{\Big(1+\frac{P}{1+aP}\Big)} \geq \log_2\Big(1+aP+\frac{1}{a}\Big)-1$ for $\frac{1}{P}<a\leq\frac{-1+\sqrt{1+4P}}{2P}$.
\begin{subequations}
\begin{eqnarray}
\log_2{\Big(1+\frac{P}{1+aP}\Big)} &\geq& \log_2\Big(1+aP+\frac{1}{a}\Big)-1\\
&\Updownarrow&\nonumber\\
P^2a^3-(1+P)a+1 =(Pa-1)(Pa^2+a-1)&\leq& 0.
\end{eqnarray}
\end{subequations}
Since $P^2a^3-(1+P)a+1=0$ has three roots of $\frac{1}{P},\frac{-1\pm \sqrt{1+4P}}{2P}$ which satisfy $0<\frac{1}{P}<\frac{-1+\sqrt{1+4P}}{2P}$ and $\frac{-1-\sqrt{1+4P}}{2P}<0$, $P^2a^3-(1+P)a+1\leq 0$ is always true for $\frac{1}{P}<a\leq\frac{-1+\sqrt{1+4P}}{2P}$.
\end{proof}
Theorem~\ref{thm:noisy} implies that IAN decoding with p2p codes performs better than the ETW scheme in noisy interference regime. Similar phenomenon can be seen in~\cite{EtTsWa08} in terms of GDOF although characterization of interference regimes in this paper is slightly different from that in~\cite{EtTsWa08}.\\
\indent
One thing to note is the interference regime in which the ETW scheme is potentially beneficial is weak interference regime. Indeed, gap between the p2p scheme and the ETW scheme becomes infinite as SNR goes to infinity as seen in~\cite{EtTsWa08}. Let us see how these two schemes compares in the practical SNR range.  
\begin{theorem}
\label{thm:20}
In weak interference regime, i.e., $\frac{-1+\sqrt{1+4P}}{2P}<a\leq 1$, we have for $P\leq P'$ where $P'=\sup\{P: f(a_1)<0, \quad a_1=\frac{5P+2-\sqrt{17P^2+12P+4}}{4P}, \quad P\geq 4 \}$,
\begin{equation}
C^{p2p}_{sym} > C^{ETW}_{sym}.
\end{equation} 
\end{theorem}
\begin{proof}
Let 
\begin{eqnarray}
g(a)&=&(1+P+aP)\Big(2+\frac{1}{a} \Big),\\
h(a)&=&\Big(1+aP+\frac{1}{a}\Big).
\end{eqnarray}
We first show that $g(a)$ and $h(a)$ are continuous function with one critical point which is the minimum for $a>0$. $g(a)$ and $h(a)$ trivially are continuous function for $a>0$. Consider $g'(a)=\frac{2Pa^2-(1+P)}{a^2}$ and $h'(a)=\frac{Pa^2-1}{a^2}$. It can be easily seen that numerators of $g'(a)$ and $h'(a)$ are quadratic function with one positive real $x$-intercept corresponding. Thus, $g(a)$ and $h(a)$ has one critical point, and it is the minimum.\\
\indent
Next, we show that $C^{p2p}_{sym} > C^{ETW}_{sym}$ if $P<4$.
\begin{enumerate}
\item $P\leq 1$\\
In this case, $INR \leq 1$ for entire weak interference regime. Therefore, the symmetric rate of the ETW scheme is $\log_2{\Big(1+\frac{P}{1+aP} \Big)}$ while that of the p2p scheme is $\frac{1}{2}\log_2(1+2P)$. From the proof of Theorem~\ref{thm:noisy}, we know that the latter is greater than the former when $a>\frac{-1+\sqrt{1+2P}}{2P}$.
\item $1<P\leq 2$\\
Note that $\frac{1}{P} \geq \frac{-1+\sqrt{1+4P}}{2P}$ when $P\leq 2$ as seen in the proof of Theorem~\ref{thm:noisy}. For $\frac{-1+\sqrt{1+4P}}{2P}<a\leq \frac{1}{P}$, we know that $\frac{1}{2}\log_2(1+2P)> \log_2{\Big(1+\frac{P}{1+aP} \Big)}$. Let us now focus on $ \frac{1}{P} <a \leq 1$. When $ \frac{1}{P} <a \leq 1$, the symmetric rate of the ETW scheme described in~\eqref{eq:etw} is no greater than $\log_2\Big(1+aP+\frac{1}{a}\Big)-1$. 
\item $2<P<4$\\
In this case, the symmetric rate of the ETW scheme is no greater than $\log_2\Big(1+aP+\frac{1}{a}\Big)-1$. We now prove the case with $2<P<4$ and complete the proof of the case with $1<P\leq 2$ at the same time. Since $h(a)$ can have at most one critical point which is the minimum in 
$\{a:\min\{\frac{1}{P}, \frac{-1+\sqrt{1+4P}}{2P} \}<a\leq 1\}$. Hence, it is sufficient to show that $\log_2\Big(1+aP+\frac{1}{a}\Big)-1$ evaluated at $a=\frac{1}{P}, \frac{-1+\sqrt{1+4P}}{2P}, 1$ is smaller than $\frac{1}{2}\log_2(1+2P)$. When $a=\frac{-1+\sqrt{1+4P}}{2P}$, we have $\log_2\Big(1+aP+\frac{1}{a}\Big)-1=\log_2 \Big( \frac{1+\sqrt{1+4P}}{2}\Big)$. 
\begin{subequations}
\begin{eqnarray}
\log_2 \Big( \frac{1+\sqrt{1+4P}}{2}\Big) &<& \frac{1}{2}\log_2(1+2P)\\
&\Updownarrow&\nonumber\\
P^2&>&0.
\end{eqnarray}
\end{subequations}
When $a=\frac{1}{P}, 1$, we have $\log_2\Big(1+aP+\frac{1}{a}\Big)-1=\log_2(2+P)-1$.
\begin{subequations}
\begin{eqnarray}
\log_2\Big(2+P \Big)-1 &<& \frac{1}{2}\log_2(1+2P)\\
&\Updownarrow&\nonumber\\
P^2-4P=P(P-4)&<&0\\
&\Updownarrow&\nonumber\\
0<&P&<4.
\end{eqnarray}
\end{subequations}
\end{enumerate}
Let us now focus on $P \geq 4$ case. Note that $a_0$ is always between $\frac{-1+\sqrt{1+4P}}{2P}$ and 1 in this case from the fact that $f(\frac{-1+\sqrt{1+4P}}{2P})=-1$ and $f(1)=P-1>0$. We first show that the symmetric rate of the p2p scheme is greater than that of the ETW scheme at the boundaries, i.e. $a=\frac{-1+\sqrt{1+4P}}{2P}, 1$. We need to show that $\log_2\Big(1+aP+\frac{1}{a}\Big)-1$ evaluated at $a=\frac{-1+\sqrt{1+4P}}{2P}$ and $\frac{1}{2}\log_2(1+P+aP)+ \frac{1}{2}\log_2\Big(2+\frac{1}{a} \Big)-1$ evaluated at $a=1$ are smaller than $\frac{1}{2}\log_2(1+2P)$. The first part is proven above while prove the case with $2<P<4$. When $a=1$, we have $\frac{1}{2}\log_2(1+P+aP)+ \frac{1}{2}\log_2\Big(2+\frac{1}{a} \Big)-1=\frac{1}{2}\log_2(1+2P)+\frac{1}{2}\log_2 \frac{3}{4}<\frac{1}{2}\log_2(1+2P)$.\\
\indent
By the aforementioned properties of $g(a)$ and $h(a)$, we know that $\frac{1}{2}\log_2(1+P+aP)+ \frac{1}{2}\log_2\Big(2+\frac{1}{a} \Big)-1$ and $\log_2\Big(1+aP+\frac{1}{a}\Big)-1$ are continuous functions of $a$ with at most one critical point corresponding to the minimum in $\{a:\frac{-1+\sqrt{1+4P}}{2P}<a\leq 1\}$. Therefore, it now suffices to show that the symmetric rate of the ETW scheme evaluated at $a=a_0$ is smaller than $\frac{1}{2}\log_2(1+2P)$, or equivalently, to show that the value of $a\leq 1$ satisfying $\frac{1}{2}\log_2(1+P+aP)+ \frac{1}{2}\log_2\Big(2+\frac{1}{a} \Big)-1=\frac{1}{2}\log_2(1+2P)$ is smaller than $a_0$ for $4 \leq P \leq P'$. Let $a \leq 1$ which satisfies $\frac{1}{2}\log_2(1+P+aP)+ \frac{1}{2}\log_2\Big(2+\frac{1}{a} \Big)-1=\frac{1}{2}\log_2(1+2P)$ be $a_1$. Then,
\begin{subequations}
\begin{eqnarray}
\frac{1}{2}\log_2(1+P+a_1P)+ \frac{1}{2}\log_2\Big(2+\frac{1}{a_1} \Big)-1&=&\frac{1}{2}\log_2(1+2P)\\
&\Updownarrow&\nonumber\\
2Pa_1^2-(5P+2)a_1+1+P&=&0\\
&\Updownarrow&\nonumber\\
a_1&=&\frac{5P+2-\sqrt{17P^2+12P+4}}{4P}. 
\end{eqnarray}
\end{subequations}
It remains to show that $f(a_1)=Pa_1^3+a_1^2-a_1-1<0$ for $4 \leq P \leq P'$. This is automatically proven if we show that $f(a_1)$ is monotonically increasing function of $P$ for $P \geq 4$. Note that  $\frac{df(a_1)}{dP}=a_1^3+(3Pa_1^2+2a_1-1)\frac{da_1}{dP}$ with $\frac{da_1}{dP}=\frac{3P+2-\sqrt{17P^2+12P+4}}{2P^2\sqrt{17P^2+12P+4}}<0$. If $3Pa_1^2+2a_1-1\leq 0$ then $\frac{df(a_1)}{dP}>0$. Let us assume that $3Pa_1^2+2a_1-1 > 0$. In this case, 
\begin{subequations}
\begin{eqnarray}
\frac{df(a_1)}{dP} &>& a_1^3+(3Pa_1^2+2a_1-1)\frac{-\sqrt{17P^2+12P+4}}{2P^2\sqrt{17P^2+12P+4}}\\
                   &=& \frac{2P^2a_1^3-3Pa_1^2-2a_1+1}{2P^2}.
\end{eqnarray}
\end{subequations}
Let $f_1(a)=2P^2a^3-3Pa^2-2a+1$. By evaluating $f'_1(a)=6P^2a^2-6Pa-2$, we can find unique positive critical point of $f_1(a)$ which corresponds to the local minimum at $a=\frac{3+\sqrt{21}}{6P}$. Since $f_1(\frac{3+\sqrt{21}}{6P})>0$ for $P> \frac{27+7\sqrt{21}}{18}$, $f_1(a)>0$ for all $a>0$ when $P>\frac{27+7\sqrt{21}}{18}$. Therefore, $\frac{df(a_1)}{dP}>0$, and hence, $f(a_1)$ is monotonically increasing function of $P$ for $P \geq 4> \frac{27+7\sqrt{21}}{18}$.  
\end{proof}
It can be easily verified that $f(a_1)<0$ when $P=100$ which implies that $P'\geq 100$. Thus, we have the following corollary.
\begin{corollary}
In weak interference regime, i.e., $\frac{-1+\sqrt{1+4P}}{2P}<a\leq 1$, we have 
\begin{equation}
C^{p2p}_{sym} > C^{ETW}_{sym},
\end{equation} 
when SNR$\leq 20$ dB.
\end{corollary}
Let us look at higher SNR range.
\begin{theorem}
\label{thm:30}
In weak interference regime, i.e., $\frac{-1+\sqrt{1+4P}}{2P}<a\leq 1$, we have for $P\leq P''$ where $P''=\sup\{P: f(a_2)<0, \quad a_2=\frac{13P+6-\sqrt{161P^2+148P+36}}{4P}, \quad P>100 \}$,
\begin{equation}
C^{p2p}_{sym} > C^{ETW}_{sym}-0.5.
\end{equation}  
\end{theorem}
\begin{proof}
From the similar reasoning to that of the proof of Theorem~\ref{thm:20}, it suffices to show that the gap between the symmetric rate of the ETW scheme evaluated at $a=a_0$ and $\frac{1}{2}\log_2(1+2P)$ is smaller than 0.5 bit, or equivalently, to show that the value of $a\leq 1$ satisfying $\frac{1}{2}\log_2(1+P+aP)+ \frac{1}{2}\log_2\Big(2+\frac{1}{a} \Big)-1=\frac{1}{2}\log_2(1+2P)+0.5$ is smaller than $a_0$ for $P\leq P''$. Let $a \leq 1$ which satisfies $\frac{1}{2}\log_2(1+P+aP)+ \frac{1}{2}\log_2\Big(2+\frac{1}{a} \Big)-1=\frac{1}{2}\log_2(1+2P)+0.5$ be $a_2$. Then,
\begin{subequations}
\begin{eqnarray}
\frac{1}{2}\log_2(1+P+a_2P)+ \frac{1}{2}\log_2\Big(2+\frac{1}{a_2} \Big)-1&=&\frac{1}{2}\log_2(1+2P)+0.5\\
&\Updownarrow&\nonumber\\
2Pa_2^2-(13P+6)a_2+1+P&=&0\\
&\Updownarrow&\nonumber\\
a_2&=&\frac{13P+6-\sqrt{161P^2+148P+36}}{4P}. 
\end{eqnarray}
\end{subequations}
It remains to show that $f(a_2)=Pa_2^3+a_2^2-a_2-1<0$ for $100<P \leq P''$. This is automatically proven if we show that $f(a_2)$ is monotonically increasing function of $P$ for $P>100$. Note that  $\frac{df(a_2)}{dP}=a_2^3+(3Pa_2^2+2a_2-1)\frac{da_2}{dP}$ with $\frac{da_2}{dP}=\frac{37P+18-3\sqrt{161P^2+148P+36}}{2P^2\sqrt{161P^2+148P+36}}<0$. If $3Pa_2^2+2a_2-1\leq 0$ then $\frac{df(a_2)}{dP}>0$. Let us assume that $3Pa_2^2+2a_2-1 > 0$. In this case, 
\begin{subequations}
\begin{eqnarray}
\frac{df(a_2)}{dP} &>& a_2^3+(3Pa_2^2+2a_2-1)\frac{-3\sqrt{161P^2+148P+36}}{2P^2\sqrt{161P^2+148P+36}}\\
                   &=& \frac{2P^2a_2^3-9Pa_2^2-6a_2+3}{2P^2}.
\end{eqnarray}
\end{subequations}
Let $f_1(a)=2P^2a^3-9Pa^2-6a+3$. By evaluating $f'_1(a)=6P^2a^2-18Pa-6$, we can find unique positive critical point of $f_1(a)$ which corresponds to the local minimum at $a=\frac{3+\sqrt{13}}{P}$. 
\begin{subequations}
\begin{eqnarray}
f_1\Big(\frac{3+\sqrt{13}}{P}\Big)&=&2P^2\Big(\frac{3+\sqrt{13}}{P}\Big)^3-9P\Big(\frac{3+\sqrt{13}}{P}\Big)^2-6\Big(\frac{3+\sqrt{13}}{P}\Big)+3\\
                                  &>&2P^2\Big(\frac{3+3}{P}\Big)^3-9P\Big(\frac{3+4}{P}\Big)^2-6\Big(\frac{3+4}{P}\Big)+3\\
                                  &=& \frac{3P-51}{P}\\
                                  &>&0.
\end{eqnarray}
\end{subequations}
Therefore, $f_1(a)>0$ for all $a>0$ when $P>100$, and hence, $\frac{df(a_2)}{dP}>0$. Consequently, $f(a_2)$ is monotonically increasing function of $P$ for $P>100$.
\end{proof}
It can be easily verified that $f(a_2)<0$ when $P=1000$ which implies that $P''\geq 1000$. Thus, we have the following corollary.
\begin{corollary}
In weak interference regime, i.e., $\frac{-1+\sqrt{1+4P}}{2P}<a\leq 1$, we have 
\begin{equation}
C^{p2p}_{sym} > C^{ETW}_{sym}-0.5,
\end{equation} 
when SNR$\leq 30$ dB.
\end{corollary}
Results obtained in this section imply that the ETW scheme performs closely to or is outperformed by the combination of IAN decoding and TDMA when SNR$\leq 30$ dB, and hence, there would be little reason for not using simpler p2p scheme for all interference regimes. This is surprising given the fact that the ETW scheme is known to be near-optimal, and implies that more complicated message splitting is possibly needed in order to significantly outperform the p2p scheme. Although there are several ways to justify the choice of simple message splitting used in the ETW scheme as discussed in~\cite{EtTsWa08}, it was also pointed out in~\cite{EtTsWa08} that this could very well be sub-optimal. One intuition of simple message splitting can be obtained from deterministic modeling~\cite{GaCo82, BrTs08} in which private information is assigned for the level under noise floor to achieve the capacity. In deterministic modeling, this choice is straightforward because no information is delivered under noise floor, but it is not necessarily the case with Gaussian IC due to randomness of noise, and the result obtained in this section can be thought as a simple example of it. The effort of finding better message splitting was given in~\cite{Sa04}, and certain value of $a$ was found for which a version of the HK scheme performs better than the p2p scheme. \\
\indent
At the start of performance evaluation, we mentioned that the sum rate of power reducing scheme is smaller than that of full power scheme. To show that, we need to look at the the sum rate characterization of the p2p scheme with asymmetric channel.  
\subsection{Two user SISO Gaussian asymmetric IC}
Consider the following SISO Gaussian asymmetric IC. Let channel inputs be $X_1$, $X_2$ and channel outputs be $Y_1$, $Y_2$,
\begin{eqnarray}
Y_1&=&\sqrt{P_1}X_1+\sqrt{a_1P_2}X_2+Z_1\\
Y_2&=&\sqrt{a_2P_1}X_1+\sqrt{P_2}X_2+Z_2,
\end{eqnarray}
where $a_1, a_2>0, P_1>P_2>0$ and $Z_1, Z_2 \sim \mathcal{CN}(0,1)$. For this channel, we would like to define interference regimes according to interference level. In noisy interference regime, $I(X_i;Y_i) \geq I(X_i;Y_j|X_j), i\neq j$ holds. In Gaussian channel, this corresponds to the range of $a_i(1+a_jP_i)\leq 1, i\neq j$. In weak interference regime,
$I(X_i;Y_i|X_j)\geq I(X_i;Y_j|X_j),   \Big(I(X_i;Y_j|X_j)>I(X_i;Y_i))\text{ or }(I(X_j;Y_i|X_i)>I(X_j;Y_j)\Big),  i\neq j$ holds. In Gaussian channel, this corresponds to the range of $a_i\leq 1, \Big(a_i(1+a_jP_i)>1 \text{ or } a_j(1+a_iP_j)>1\Big), i\neq j$. In mixed interference regime, $I(X_i;Y_i|X_j)\geq I(X_i;Y_j|X_j),  I(X_j;Y_j|X_i)< I(X_j;Y_i|X_i),   i\neq j$ holds. In Gaussian channel, this corresponds to the range of $a_i> 1,  a_j\leq 1,  i\neq j$. Let us further divide mixed interference regime into two sub-regimes. In direct-link-limited mixed interference regime, $I(X_i;Y_i|X_j)\geq I(X_i;Y_j|X_j),  I(X_j;Y_j|X_i)< I(X_j;Y_i|X_i), I(X_j;Y_i)\geq I(X_j;Y_j), i\neq j$ holds. In Gaussian channel, this corresponds to the range of $a_i>1, a_j\leq 1, a_i(1+a_jP_i)\geq 1+P_i, i\neq j$. In cross-link-limited mixed interference regime, $I(X_i;Y_i|X_j)\geq I(X_i;Y_j|X_j),  I(X_j;Y_j|X_i)< I(X_j;Y_i|X_i), I(X_j;Y_i)< I(X_j;Y_j), i\neq j$ holds. In Gaussian channel, this corresponds to the range of $a_i>1, a_j\leq 1, a_i(1+a_jP_i)< 1+P_i, i\neq j$. In mixed interference regime, receiver $i$ such that $a_i>1$ sees the better MAC channel than receiver $j$. To achieve better sum rate, receiver $j$ is forced to perform IAN decoding and decodability of user $j$'s message becomes the limiting factor. In direct-link-limited mixed interference regime, such decodability of the direct link becomes the limiting factor, and vice versa in cross-link-limited mixed interference regime. In strong interference regime, $I(X_i;Y_i|X_j)< I(X_i;Y_j|X_j), i\neq j$ holds. In Gaussian channel, this corresponds to the range of $a_i>1$.\\
\indent
The capacity region of p2p-capacity-achieving codes is given as $\cap_i \cup_j \mathcal{C}_{i,j}$~\cite{BaGaTs11}, where
\begin{subequations}
\begin{eqnarray}
\mathcal{C}_{i,0}&=&\left\{R_i: R_i < I(X_i;Y_i)\right\},\\
\mathcal{C}_{i,1}&=&\{(R_i,R_j): R_i < I(X_i;Y_i|X_j), R_j < I(X_j;Y_i|X_i),\nonumber\\
&& \qquad \qquad \quad R_i+R_j<I(X_i,X_j;Y_i)\}.
\end{eqnarray}
\end{subequations}
It has been shown the capacity region of p2p-capacity-achieving codes is equal to the capacity region in strong interference regime~\cite{Ca75, Sa78, HaKo81}. Let us now find the maximum sum rate with p2p-capacity-achieving codes. We only consider noisy, weak and mixed interference regimes. 
\begin{itemize}
\item \textbf{The maximum sum rate of p2p-capacity-achieving codes}
\begin{enumerate}
\item Noisy interference regime ($a_i(1+a_jP_i)\leq 1, i\neq j$)
\begin{equation} 
 C^{p2p}_{sum}= \sum_i \log_2{\Big(1+\frac{P_i}{1+a_iP_j} \Big)}.
 \end{equation}
\item Weak interference regime \bigg($a_i\leq 1, \Big(a_i(1+a_jP_i)>1 \text{ or } a_j(1+a_iP_j)>1\Big), i\neq j$\bigg)
\begin{equation}
C^{p2p}_{sum}=\max \left\{ \log_2(1+P_1+a_1P_2), \log_2(1+a_2P_1+P_2) \right \}
\end{equation}
\item Mixed interference regime ($a_i >1, a_j \leq 1, i\neq j$)
\begin{equation}
C^{p2p}_{sum}=\begin{cases}    \log_2 (1+P_i+a_iP_j), &\quad  a_i(1+a_jP_i)< 1+P_i \\
                               \log_2 (1+\frac{P_j}{1+a_jP_i})+\log_2(1+P_i) &\quad  a_i(1+a_jP_i)\geq 1+P_i \end{cases}.
\end{equation}
\end{enumerate}
We can see that in weak and mixed interference regimes, the MAC sum rate bound at the better receiver can be achieved except for the direct-link-limited mixed interference regime. This is because the worse receiver is forced to use IAN decoding. Hence, the rate of user corresponding to the better receiver can be large enough achieve its receiver's sum rate bound. In direct-link-limited mixed interference regime, however, this cannot be achieved because decodability of message of the worse receiver's user at the direct link is too low. We now state the result saying that power reducing scheme is sub-optimal to full power scheme for two user SISO Gaussian symmetric IC. To do that we consider the case with $a_1=a_2=a$, and hence, there is no mixed interference regime. 
\begin{theorem}
Given $0<a\leq 1$, the maximum sum rate obtained by p2p-capacity-achieving codes with $P_1=P_2=P$ is always no smaller than with $P_1,P'_2$ such that $P'_2\leq P$. 
\end{theorem}
\begin{proof}
If the system was in noisy interference regime with $P_1, P_2$, then it remains in noisy interference regime with $P_1, P'_2$ from the definition of noisy interference regime. Then, it is sufficient to show that $g(p)=(1+\frac{p}{1+aq})(1+\frac{q}{1+ap})$ is increasing function of $p>0$ when $a(1+ap)\leq 1, a(1+aq)\leq 1 $, i.e, $\frac{d g(p)}{d p}>0$.
\begin{subequations}
\begin{eqnarray}
\frac{d g(p)}{d p}&=& \frac{1}{1+aq}\left(1+\frac{q}{1+ap}\right)-\frac{a}{(1+ap)^2}\left(1+\frac{p}{1+aq}\right)\\
                                  &=& \frac{1}{(1+ap)(1+aq)}\left( 1+ap+q-\frac{a(1+aq+p)}{1+ap}    \right)\\
                                  &\geq& \frac{1}{(1+ap)(1+aq)}\left( \frac{1+aq}{a}-\frac{a(1+aq+p)}{1+ap}    \right)\\
                                  &\geq& \frac{a^2pq}{a(1+ap)^2(1+aq)}\\
                                  &>&0.                               
\end{eqnarray}
\end{subequations}
If the system was in weak interference regime with $P_1, P_2$, then it should remain in weak interference regime with $P_1, P'_2$. In this case, $P_1, P_2$ trivially has larger sum rate than $P_1, P'_2$. 
\end{proof}
\end{itemize}

\section{$K$ user SISO Gaussian symmetric IC}
\label{sec:k}
Let us define $K$ user SISO Gaussian symmetric IC with channel inputs $X_1$,..., $X_K$ and channel outputs $Y_1$,..., $Y_K$ as follows.
\begin{eqnarray}
Y_i&=&\sqrt{P}X_i+\sqrt{aP}\sum_{j \neq i}X_j+Z_i, \quad i=1,...,K,
\end{eqnarray}
where $a,P>0$ and $Z_i \sim \mathcal{CN}(0,1)$. As mentioned earlier, signal level alignment is shown to achieve GDOF of this channel, and HK-like scheme would be sub-optimal. This is because of violation of ``interference decodability'' in this channel, which can be easily understood via deterministic modeling~\cite{GaCo82, GoJa11}. This property makes HK-like scheme which requires decoding of ``all'' common information be sub-optimal, and it is also a reason why signal level alignment which ensures decodability of ``sum'' of interfering signals can achieve GDOF. Nevertheless, investigating performances of simple strategies still can be meaningful given the fact that these are considerably easier to implement.\\
\indent
Suppose there is no coordination among transmitters and no message splitting at each transmitter. In this case, each receiver has an option of decoding $k=1,...,K$ messages while treating remaining $(K-k)$ messages as noise. Note that each receiver must decode the intended message. As in two user case, we are interested in the symmetric rate as a performance metric. To find the symmetric rate, we need to evaluate achievable region of this scheme. Consider now the capacity region $\mathcal{C}$ of Gaussian-p2p codes defined in equation (1) of~\cite{BaGaTs11}. This region can be relatively easily analyzed, and hence we would like to focus in this region. As seen in~\cite{BaGaTs11}, however, $\mathcal{C}$ cannot be thought as the capacity region of p2p-capacity-achieving codes since it is possible to gain benefit even without message splitting by coordination in $K$ user case. This is why MAC-capacity-achieving codes are defined in~\cite{BaGaTs11}, and it turns out that $\mathcal{C}$ is the capacity region with MAC-capacity-achieving codes. Simply speaking, $\mathcal{C}$ can be thought as the capacity region with no coordination among transmitters and no message splitting at each transmitter which we are interested in, and we will call this region as capacity region of Gaussian-p2p codes. \\
\indent
Due to symmetry, the symmetric rate is obtained by evaluation on one receiver. Let us consider receiver 1. Let $\mathcal{S}$ be some subset of $\{2,...,K\}$, and let $X_\mathcal{S}$ be the vector of transmitted signals $X_i$ such that $i\in \mathcal{S}$, and $R_\mathcal{S}$ be the corresponding vector of rates. Then, the symmetric rate $\tilde{C}^{p2p}_{sym}$ with Gaussian-p2p codes is given as 
\begin{equation}
\tilde{C}^{p2p}_{sym} =\max_\mathcal{S} \{\tilde{C}^{p2p}_{sym,\mathcal{S} } \},
\end{equation}
where
\begin{equation}
\tilde{C}^{p2p}_{sym,\mathcal{S}}=\min_{ \mathcal{T} \subseteq \mathcal{S}  } \Big \{   \frac{1}{|\mathcal{T}|+1} I(X_1,X_\mathcal{T}; Y_1 | X_{\mathcal{S}\backslash \mathcal{T}})     \Big \}
\end{equation}
We will show that the symmetric rate $\tilde{C}^{p2p}_{sym}$ is the same as $\hat{C}^{p2p}_{sym}=\max_{\mathcal{S}=\emptyset, \{2,...,K\}} \{\tilde{C}^{p2p}_{sym,\mathcal{S} } \}$. This implies that the symmetric rate is achieved by decoding all interference messages or treating all interference messages as noise. To do that, we first show that $\tilde{C}^{p2p}_{sym,\mathcal{S}}$ is the same as $\hat{C}^{p2p}_{sym,\mathcal{S}}=\min \Big \{I(X_1;Y_1|X_\mathcal{S}), \frac{1}{|\mathcal{S}|+1}I(X_1,X_{\mathcal{S}};Y_1) \Big \}$. This implies that the active bound for the symmetric rate is always either the individual rate bound or the total sum rate bound. 
\begin{lemma}
\label{lem:kbound}
For every $\mathcal{S} \subseteq \{2,...,K\}$, we have
\begin{equation}
\tilde{C}^{p2p}_{sym,\mathcal{S}}=\hat{C}^{p2p}_{sym,\mathcal{S}}.
\end{equation}
\end{lemma} 
\begin{proof}
Given $\mathcal{S}$, it is sufficient to show that $\frac{1}{1+|\mathcal{T}|}I(X_1,X_\mathcal{T}; Y_1 | X_{\mathcal{S}\backslash \mathcal{T}})\geq \hat{C}^{p2p}_{sym,\mathcal{S}}$ for all $\mathcal{T} \subseteq \mathcal{S}$. We prove this by considering two cases in which $\hat{C}^{p2p}_{sym,\mathcal{S}}=I(X_1;Y_1|X_\mathcal{S})$ or $\hat{C}^{p2p}_{sym,\mathcal{S}}=\frac{1}{|\mathcal{S}|+1}I(X_1,X_{\mathcal{S}};Y_1)$. \\
\indent
First, assume that $(|\mathcal{S}|+1)\times I(X_1;Y_1|X_\mathcal{S})>I(X_1,X_{\mathcal{S}};Y_1)$. Note that  $\frac{1}{|\mathcal{T}|}I(X_\mathcal{T}; Y_1 | X_{\mathcal{S}\backslash \mathcal{T}})\geq  \frac{1}{|\mathcal{S}\backslash \mathcal{T}|}I( X_{\mathcal{S}\backslash \mathcal{T}}; Y_1)$ for all $\mathcal{T} \subseteq \mathcal{S}$. Therefore, the assumption implies that $|\mathcal{S}\backslash \mathcal{T}|\times I(X_1;Y_1|X_\mathcal{S})> I(X_{\mathcal{S}\backslash \mathcal{T}};Y_1)$ for all $\mathcal{T} \subseteq \mathcal{S}$. Consequently, we have $\frac{1}{1+|\mathcal{T}|}I(X_1,X_\mathcal{T}; Y_1 | X_{\mathcal{S}\backslash \mathcal{T}}) \geq \frac{1}{1+|\mathcal{S}|}\Big (I(X_1,X_\mathcal{T}; Y_1 | X_{\mathcal{S}\backslash \mathcal{T}})+ I( X_{\mathcal{S}\backslash \mathcal{T}}; Y_1)\Big )=\hat{C}^{p2p}_{sym,\mathcal{S}}$ for all $\mathcal{T} \subseteq \mathcal{S}$.\\
\indent
Now assume that $(|\mathcal{S}|+1)\times I(X_1;Y_1|X_\mathcal{S})\leq I(X_1,X_{\mathcal{S}};Y_1)$. As in the previous case, we have $\frac{1}{|\mathcal{T}|}I(X_\mathcal{T}; Y_1 | X_{\mathcal{S}\backslash \mathcal{T}})\geq  \frac{1}{|\mathcal{S}\backslash \mathcal{T}|}I( X_{\mathcal{S}\backslash \mathcal{T}}; Y_1)$ for all $\mathcal{T} \subseteq \mathcal{S}$. Threrefore, the assumption implies that $| \mathcal{T}|\times I(X_1;Y_1|X_\mathcal{S})\leq I(X_\mathcal{T}; Y_1 | X_{\mathcal{S}\backslash \mathcal{T}})$ for all $\mathcal{T} \subseteq \mathcal{S}$, which means that $\frac{1}{1+|\mathcal{T}|}I(X_1,X_\mathcal{T}; Y_1 | X_{\mathcal{S}\backslash \mathcal{T}}) \geq  I(X_1;Y_1|X_\mathcal{S})=\hat{C}^{p2p}_{sym,\mathcal{S}}$ for all $\mathcal{T} \subseteq \mathcal{S}$.
\end{proof} 
Using the above lemma we show that the symmetric rate is obtained by IAN decoding or joint decoding of all messages.
\begin{theorem}
\begin{equation}
\tilde{C}^{p2p}_{sym}=\hat{C}^{p2p}_{sym}.
\end{equation}
\end{theorem}
\begin{proof}
Because of Lemma~\ref{lem:kbound}, it suffices to show that $\max_\mathcal{S} \{\hat{C}^{p2p}_{sym,\mathcal{S} } \}=\hat{C}^{p2p}_{sym}$.\\
\indent
First, assume that $\hat{C}^{p2p}_{sym,\mathcal{S}=\{2,...,K\}} < \hat{C}^{p2p}_{sym,\mathcal{S}=\emptyset}$. This implies that $| \mathcal{S}'|\times I(X_1;Y_1)> I(X_{\mathcal{S}'}; Y_1 | X_1)$ for all $\mathcal{S}' \subseteq \{2,...,K\}$. Hence, $\hat{C}^{p2p}_{sym,\mathcal{S}=\mathcal{S}' }\leq \frac{1}{|\mathcal{S}'|+1}I(X_1,X_{\mathcal{S}'};Y_1) < \hat{C}^{p2p}_{sym,\mathcal{S}=\emptyset}$ for all $\mathcal{S}' \subseteq \{2,...,K\}$.\\
\indent
Next, assume that $\hat{C}^{p2p}_{sym,\mathcal{S}=\{2,...,K\}} \geq \hat{C}^{p2p}_{sym,\mathcal{S}=\emptyset}$. If $\hat{C}^{p2p}_{sym,\mathcal{S}=\{2,...,K\}}=I(X_1;Y_1|X_\mathcal{S})$, then $\hat{C}^{p2p}_{sym}$ is the same as the rate achieved by no interference which cannot be exceeded.\\
\indent
Now assume that $\hat{C}^{p2p}_{sym,\mathcal{S}=\{2,...,K\}}\geq \hat{C}^{p2p}_{sym,\mathcal{S}=\emptyset}$, and $\hat{C}^{p2p}_{sym,\mathcal{S}=\{2,...,K\}}=\frac{1}{|\mathcal{S}|+1}I(X_1,X_{\mathcal{S}};Y_1)$. The former implies that $| \mathcal{S}'|\times I(X_1;Y_1)\leq I(X_{\mathcal{S}'}; Y_1 | X_1, X_{\{2,...,K\}\backslash \mathcal{S}'})$ for all $\mathcal{S}' \subseteq \{2,...,K\}$. Therefore, $\frac{1}{|\mathcal{S}'|+1}I(X_1,X_{\mathcal{S}'};Y_1) \leq  \frac{1}{|\{2,...,K\}|+1}I(X_1,X_{\{2,...,K\}};Y_1)$ for all $\mathcal{S}' \subseteq \{2,...,K\}$. Hence, $\hat{C}^{p2p}_{sym,\mathcal{S}=\mathcal{S}' }\leq \frac{1}{|\mathcal{S}'|+1}I(X_1,X_{\mathcal{S}'};Y_1) \leq \hat{C}^{p2p}_{sym,\mathcal{S}=\{2,...,K\}}$.
\end{proof}
We now define four interference regimes for this channel. Noisy interference regime satisfies $ I(X_1;Y_1)>\frac{1}{K} I(X_1,...,X_K;Y_1)$. Weak interference regime satisfies $  I(X_1;Y_1)\leq \frac{1}{K} I(X_1,...,X_K;Y_1)$ and $I(X_1;Y_1|X_2,...,X_K)>I(X_2;Y_1|X_1,X_3,...,X_K)$. Strong interference regime satisfies
$ \frac{1}{K} I(X_1,...,X_K;Y_1)<I(X_1;Y_1|X_2,...,X_K)\leq I(X_2;Y_1|X_1,X_3,...,X_K)$. Very strong interference regime satisfies $ \frac{1}{K} I(X_1,...,X_K;Y_1)\geq I(X_1;Y_1|X_2,...,X_K)$. Note that we know the capacity for very strong interference regime. \\
\indent
We can characterize the symmetric rate of Gaussian-p2p codes for each regime.
\begin{itemize}
\item \textbf{The symmetric rate of Gaussian-p2p codes}
\begin{enumerate}
\item Noisy interference regime $\Big((1+(K-1)aP+P)^{K-1}>(1+(K-1)aP)^K\Big)$
\begin{equation} 
 \tilde{C}^{p2p}_{sym}= \log_2{\Big(1+\frac{P}{1+(K-1)aP} \Big)}.
 \end{equation}
\item Weak interference regime $\Big((1+(K-1)aP+P)^{K-1}\leq (1+(K-1)aP)^K, a<1     \Big)$
\begin{equation}
\tilde{C}^{p2p}_{sym}=  \frac{1}{K}\log_2(1+P+(K-1)aP).
\end{equation}
\item Strong interference regime $\Big((1+(K-1)aP+P) < (1+P)^K, a\geq 1     \Big)$
\begin{equation}
\tilde{C}^{p2p}_{sym}=  \frac{1}{K}\log_2(1+P+(K-1)aP).
\end{equation}
\item Very strong interference regime $\Big((1+(K-1)aP+P) \geq (1+P)^K\Big)$
\begin{equation}
\tilde{C}^{p2p}_{sym}=\log_2 (1+P).
\end{equation}
\end{enumerate}
\end{itemize}
Note that TDMA achieves the symmetric rate $\tilde{C}^{TDMA}_{sym}=\frac{1}{K}\log_2 (1+KP)$. We would still call the combined scheme of TDMA and Gaussian-p2p codes as ``p2p scheme'' as in Section~\ref{sec:sym}. To compare performance, we need to characterize the symmetric rate of the ETW scheme. Let be $\mathcal{K}_k$ the set of every subset of $\{1,...,K\}$ with cardinality $k$. Using sub-optimal decoding mentioned in Section~\ref{sec:sym}, achievable region of the ETW scheme for $\frac{1}{P} < a$ is given as $\mathcal{C}^{K}_{ETW}=\Big \{(R_{c_1},..., R_{c_K}, R_{p_1},..., R_{p_K}): R_{p_i}<\log_2(1+\frac{1}{2a}), \sum_{i \in K_k} R_{c_i} < \min \{ \log_2 (1+\frac{k(aP-1)}{K+1/a}), \log_2 (1+\frac{(k-1)(aP-1)+P-1/a}{K+1/a})          \} \text{ for } K_k \in \mathcal{K}_k \text{ with } k\neq K,   \sum_{i=1}^K R_{c_i}< \log_2(1+\frac{(K-1)(aP-1)+P-1/a}{K+1/a})\Big \}$, where $R_{c_i}$ is the rate of user $i$'s common message, and $R_{p_i}$ is the rate of user $i$'s private message. We will characterize the symmetric rate of the ETW scheme with sub-optimal decoding in closed form. Note that the symmetric rate of the ETW scheme for two user case is not characterized for strong interference regime in~\cite{EtTsWa08} because the p2p scheme achieves the capacity. In $K$ user case, however, this is not the case.
\begin{theorem}
The symmetric rate of the ETW scheme is given as
\begin{equation}
\tilde{C}^{ETW}_{sym}=
\begin{cases}
\log_2{\Big(1+\frac{P}{1+(K-1)aP} \Big)},& \quad a\leq \frac{1}{P}\\
\log_2(1+\frac{1}{Ka})+ \min \bigg\{ \frac{1}{K-1}\log_2 \Big(1+\frac{(K-1)(aP-1)}{K+1/a}\Big), \frac{1}{K}\log_2\Big(1+\frac{(K-1)(aP-1)+P-1/a}{K+1/a}\Big)            \bigg\}, & \quad \frac{1}{P}<a<1\\
\log_2(1+\frac{1}{Ka})+ \min \bigg\{ \log_2 \Big(1+\frac{P-1/a}{K+1/a}\Big), \frac{1}{K}\log_2\Big(1+\frac{(K-1)(aP-1)+P-1/a}{K+1/a}\Big)            \bigg\}, & \quad a\geq 1.
\end{cases}
\end{equation}
\end{theorem}
\begin{proof}
All messages are private if $a<1/P$, and this case is trivial. For other cases, we can rewrite the achievable region of the ETW scheme as $\mathcal{C}^{K}_{ETW}=\Big \{(R_{c_1},..., R_{c_K}, R_{p_1},..., R_{p_K}): R_{p_i}<I(X_{p_1}; Y_1|X_{c_1},...,X_{c_K}), \sum_{i \in K_k} R_{c_i} < \min \{ I(X_{c_2},...,X_{c_{k+1}};Y_1|X_{c_1},X_{c_{k+2}},...,X_{c_{K}}), I(X_{c_1}, X_{c_2},...,X_{c_{k}};Y_1|X_{c_{k+1}},...,X_{c_{K}})\} \text{ for } K_k \in \mathcal{K}_k \text{ with } k\neq K,   \sum_{i=1}^K R_{c_i}< I(X_{c_1},...,X_{c_K};Y_1)\Big \}$.\\
\indent
Let us now consider the case with $1/P \leq a <1$. Since $a<1$, we have $I(X_{c_{k+1}};Y_1|X_{c_1},X_{c_{k+2}},...,X_{c_{K}})< I(X_{c_1};Y_1|X_{c_{k+1}},...,X_{c_{K}})$ for all $k$. Therefore,
\begin{subequations}
\begin{eqnarray}
\sum_{i \in K_k} R_{c_i}& < & \min \Big \{ I(X_{c_2},...,X_{c_{k+1}};Y_1|X_{c_1},X_{c_{k+2}},...,X_{c_{K}}), I(X_{c_1}, X_{c_2},...,X_{c_{k}};Y_1|X_{c_{k+1}},...,X_{c_{K}}) \Big \}\\
&=&I(X_{c_2},...,X_{c_{k+1}};Y_1|X_{c_1},X_{c_{k+2}},...,X_{c_{K}}).
\end{eqnarray}
\end{subequations}
It is easy to see that $\frac{1}{k-1}I(X_{c_2},...,X_{c_{k}};Y_1|X_{c_1},X_{c_{k+1}},...,X_{c_{K}})\geq \frac{1}{k}I(X_{c_2},...,X_{c_{k+1}};Y_1|X_{c_1},X_{c_{k+2}},...,X_{c_{K}})$. Therefore, $\tilde{C}^{ETW}_{sym}=I(X_{p_1}; Y_1|X_{c_1},...,X_{c_K})+ \min \bigg\{ \frac{1}{K-1}I(X_{c_2},...,X_{c_{K}};Y_1|X_{c_1}), \frac{1}{K}I(X_{c_1},...,X_{c_K};Y_1)            \bigg\}$.\\
\indent
Consider now the case with $a \geq 1$. Since $a\geq 1$, we have $I(X_{c_{k+1}};Y_1|X_{c_1},X_{c_{k+2}},...,X_{c_{K}})\geq I(X_{c_1};Y_1|X_{c_{k+1}},...,X_{c_{K}})$ for all $k$. Therefore,
\begin{subequations}
\begin{eqnarray}
\sum_{i \in K_k} R_{c_i}& < & \min \Big \{ I(X_{c_2},...,X_{c_{k+1}};Y_1|X_{c_1},X_{c_{k+2}},...,X_{c_{K}}), I(X_{c_1}, X_{c_2},...,X_{c_{k}};Y_1|X_{c_{k+1}},...,X_{c_{K}})\Big\}\\
&=&I(X_{c_1}, X_{c_2},...,X_{c_{k}};Y_1|X_{c_{k+1}},...,X_{c_{K}}).
\end{eqnarray}
\end{subequations}
It is easy to see that $\frac{1}{k-1}I(X_{c_1}, X_{c_2},...,X_{c_{k-1}};Y_1|X_{c_{k}},...,X_{c_{K}})  \geq \frac{1}{k}I(X_{c_1}, X_{c_2},...,X_{c_{k}};Y_1|X_{c_{k+1}},...,X_{c_{K}})$. Therefore, $\tilde{C}^{ETW}_{sym}=I(X_{p_1}; Y_1|X_{c_1},...,X_{c_K})+ \min \bigg\{ I(X_{c_1};Y_1|X_{c_2},..., X_{c_K} ), \frac{1}{K}I(X_{c_1},...,X_{c_K};Y_1))            \bigg\}$.
\end{proof}
As in Section~\ref{sec:sym}, we would like to compare performances of the p2p scheme and the ETW scheme for weak interference regime. Unfortunately, complete analysis like in Section~\ref{sec:sym} is extremely difficult for more than two users. Because of that, we will restrict our attention to $K=3$ with the approximated symmetric rate. Let us define the approximated symmetric rates as 
\begin{eqnarray}
\hat{C}^{TDMA}_{sym}&=&\frac{1}{K}\log_2(KP),\\
\hat{C}^{ETW}_{sym}&=&\log_2(1+\frac{1}{Ka})\nonumber \\
&&+ \min \bigg\{ \frac{1}{K-1}\log_2 \Big(\frac{(K-1)(aP-1)}{K+1/a}\Big),\nonumber\\
&& \qquad \qquad \frac{1}{K}\log_2\Big(\frac{K(P-1/a)}{K+1/a}\Big)            \bigg\}.
\end{eqnarray}
We now compare performances in terms of these approximated symmetric rates.
\begin{theorem}
With $K=3$, we have for $\frac{-24+9\sqrt{10}}{26}<P<\frac{142389}{2048}$ in $a<1$,
\begin{equation}
\hat{C}^{ETW}_{sym}<\hat{C}^{TDMA}_{sym}.
\end{equation}
\end{theorem}
\begin{proof}
With $K=3$, the approximated symmetric rate of the ETW scheme becomes
\begin{equation}
\hat{C}^{ETW}_{sym}= \min \bigg\{ \frac{1}{2}\log_2\Big(\frac{2}{3}(aP-1)(1+\frac{1}{3a})\Big)       ,  \frac{1}{3}\log_2\Big((P-\frac{1}{a})(1+\frac{1}{3a})^2       \Big) \bigg\}.
\end{equation}
First, we would like to find range of $a$ which satisfies the following.
\begin{subequations}
\begin{eqnarray}
\frac{1}{2}\log_2\Big(\frac{2}{3}(aP-1)(1+\frac{1}{3a})\Big) &<& \frac{1}{3}\log_2\Big((P-\frac{1}{a})(1+\frac{1}{3a})^2 \Big)\\
&\Updownarrow&\nonumber\\
8Pa^4-8a^3-27a-9 &<& 0.
\end{eqnarray}
\end{subequations}
Let $g_1(a)=8Pa^4-8a^3-27a-9$. Since $g_1'(a)=32Pa^3-24a^2-27$ has a critical point at $a=0$ which corresponds to the local maximum and $g_1'(0)<0$, we know that $g_1'(a)$ has one $x-$intercept which means that $g_1(a)$ has one critical point which corresponds to the local minimum. Since $g_1(0)<0$, we know that $g_1(a)<0$ if $0<a<a_1$, and $fg_1(a)\geq 0$ otherwise, where $a_1$ is the unique positive root of $g_1(a)=0$. In other words, $\hat{C}^{ETW}_{sym}$ equals to $\frac{1}{2}\log_2\Big(\frac{2}{3}(aP-1)(1+\frac{1}{3a})\Big)$ if $0<a<a_1$, and it equals to $\frac{1}{3}\log_2\Big((P-\frac{1}{a})(1+\frac{1}{3a})^2       \Big)$ otherwise.\\
\indent
Consider now the range of $a$ which satisfies the following.
\begin{subequations}
\begin{eqnarray}
\frac{1}{3}\log_2\Big((P-\frac{1}{a})(1+\frac{1}{3a})^2 \Big) &<& \frac{1}{3}\log_2(3P) \\
&\Updownarrow&\nonumber\\
18Pa^3-(6P-9)a^2-(P-6)a+1&>& 0.
\end{eqnarray}
\end{subequations}
Let $g_2(a)=18Pa^3-(6P-9)a^2-(P-6)a+1$. Note that the largest critical point of $g_2(a)$ which corresponds to the local minimum is $a_2=\frac{2P-3+\sqrt{10P^2-48P+9}}{18P}$. Since $a_2<\frac{4}{9}$ if $P>\frac{-24+9\sqrt{10}}{26}$, and $g_2(\frac{4}{9})=\frac{-4P+441}{81}>0$ if $P<\frac{441}{4}$, we know that $g_2(a)>0$ for $a>\frac{4}{9}$ if $\frac{-24+9\sqrt{10}}{26}<P<\frac{441}{4}$. Therefore, we can show that $\hat{C}^{ETW}_{sym}$ for $a\geq a_1$ by showing that $\frac{4}{9}<a_1$, i.e., $g_1(\frac{4}{9})<0$. We can easily see that $g_1(\frac{4}{9})<0$ if $P<\frac{142389}{2048}$. Furthermore, we have $g'_3(a)=P+\frac{1}{3a^2}>0$ for $g_3(a)= (aP-1)(1+\frac{1}{3a})$, which means $g_3(a)$ is monotonically increasing. Hence, we get $ $ for $a<a_1$ if it is true at $a=a_1$. Therefore, we have the claim.
\end{proof}
The above result shows that the trend is similar to two user case even with more than two users, although the result obtained here is weaker and more limited due to analytical difficulty. Note that the p2p scheme is not shown to achieve the capacity for strong interference regime unlike two user case. Hence, it would be worthwhile to compare performances in this regime, but we expect to see similar trend to two user case.
\section{Concluding remarks}
\label{sec:con}
In this paper, we investigated performances of simple transmission schemes for interference channel. It turns out that very simple transmission scheme even without message splitting can be quite good. Although the TDMA scheme is mainly compared with the ETW scheme for weak interference regime, we should note that this was because the p2p scheme is already known to be capacity achieving for strong and very strong interference regimes with two users. In other words, it is important to have the receiver structure which is capable of decoding multiple messages for strong interference, but simple change in scheduling could be good enough at the transmitter. We may also interpret this result in the way that there needs to be more careful consideration of message splitting for practical purpose as considered in~\cite{DaYu11}.\\
\indent
In addition to performance comparison, we characterized interference regimes and the maximum rate of the p2p scheme in this paper. These characterizations provides several insightful explanations which help understanding of effect of interference.\\
\indent
We may also think of $K$ user asymmetric IC to investigate the performance of the ETW-like scheme and the p2p scheme. To enable message splitting in this case, each transmitter must have multiple message splitting. Progress in this direction has been already made in~\cite{GoTaWa11} in the name of partial group decoding. \\
\indent
Performance of the p2p scheme in MIMO IC was investigated in~\cite{ZhCi11} where the transmit signal covariance optimization is performed. It would be interesting to compare performance of the p2p scheme with that of GDOF optimal schemes in~\cite{KaVa11-1,PaBlTa08}. 


\bibliographystyle{IEEEtran}
\bibliography{bae}
\end{document}